\documentclass[12pt]{amsart}
\usepackage{amsfonts,color,latexsym,amssymb,amsmath,epsfig,rotating,array}

\newtheorem{dfn}{Definition}
\newtheorem{thm}[dfn]{Theorem}

\newtheorem{lem}[dfn]{Lemma}
\newtheorem{prop}[dfn]{Proposition}

\newtheorem{ex}[dfn]{Example}
\newtheorem{rem}[dfn]{Remark}

\theoremstyle{definition}
\newtheorem*{MDZH}{Minimalist Dedekind Zero Hypothesis (MDZH)}
\newtheorem*{MRH}{Modular Root Hypothesis (MRH)}

\newcommand{\bO}{\mathbf{O}}
\newcommand{\fa}{\mathfrak{a}}
\newcommand{\fp}{\mathfrak{p}}
\newcommand{\eps}{\varepsilon}
\newcommand{\pp}{\mathbf{P}}
\newcommand{\vp}{\mathbf{VP}}
\newcommand{\vnp}{\mathbf{VNP}}
\newcommand{\np}{\mathbf{NP}}
\newcommand{\conp}{\mathbf{coNP}}
\newcommand{\am}{\mathbf{AM}}
\newcommand{\bpp}{\mathbf{BPP}}
\newcommand{\pspa}{\mathbf{PSPACE}}
\newcommand{\dia}{$\diamond$}

\newcommand{\C}{\mathbb{C}}
\newcommand{\F}{\mathbb{F}}
\newcommand{\N}{\mathbb{N}}
\newcommand{\Q}{\mathbb{Q}}
\newcommand{\R}{\mathbb{R}}

\newcommand{\Z}{\mathbb{Z}}

\newcommand{\cE}{\mathcal{E}}

\newcommand{\cO}{\mathcal{O}}

\newcommand{\lcm}{\operatorname{lcm}}
\newcommand{\la}{\langle}
\newcommand{\ra}{\rangle}

\begin{document}
\title{Dedekind Zeta Zeroes and Faster Complex Dimension Computation}

\author{J.\ Maurice Rojas}
\email{rojas@math.tamu.edu} 
\address{Texas A\&{}M University, TAMU 3368, College Station, 
TX 77843-3368} 
\thanks{Partially supported by NSF grants DMS-1460766
and CCF-1409020. }

\author{Yuyu Zhu}
\email{zhuyuyu@math.tamu.edu}

\begin{abstract}
Thanks to earlier work of Koiran, it is known that the truth of the 
Generalized Riemann Hypothesis (GRH) implies that the dimension 
of algebraic sets over the complex numbers can be determined 
within the polynomial-hierarchy. The truth of GRH thus provides a 
direct connection between a concrete algebraic geometry problem and 
the $\pp$ vs.\ $\np$ Problem, in a radically different direction 
from the geometric complexity theory approach to $\vp$ vs.\ $\vnp$. We 
explore more plausible hypotheses yielding the same speed-up. One minimalist 
hypothesis we derive involves improving the error term (as a function 
of the degree, coefficient height, and $x$) on 
the fraction of primes $p\!\leq\!x$ for which a univariate polynomial 
has roots mod $p$. A second minimalist hypothesis involves sharpening 
current zero-free regions for Dedekind zeta functions. Both our 
hypotheses allow failures of GRH but still enable complex dimension 
computation in the polynomial hierarchy.  
\end{abstract}


\maketitle

\section{Introduction} 
The subtlety of computational complexity in algebraic 
geometry persists in some of its most basic problems. For instance, 
	let {\tt FEAS$_\C$} denote the
problem of deciding whether an input polynomial system\\ 
\mbox{}\hfill  
$F\!\in\!\bigcup\limits_{k,n\in \N} 
(\Z[x_1,\ldots,x_n])^k$\hfill\mbox{}\\ 
has a complex root. While the 
implication $\text{{\tt FEAS}}_\C\!\in\!\pp\Longrightarrow 
\pp\!=\!\np$ has long been known, the inverse implication  
$\text{{\tt FEAS}}_{\C} \not \in \pp \Longrightarrow \pp\!\neq\!\np$ 
remains unknown.  
Proving the implication $\text{{\tt FEAS}}_{\C}
\not \in \pp \Longrightarrow \pp\!\neq\!\np$ 
would shed new 
light on the $\pp$ vs.\ $\np$ Problem, and may be easier 
than attempting to prove 
the complexity lower bound $\text{{\tt FEAS}}_\C 
\not \in \pp$ (whose truth is still unknown). 

Detecting complex roots is the $D\!=\!0$ case of the following more 
general problem: 
\begin{quote} 
{\tt DIM}$_\C$: Given $\displaystyle{(D,F)\!\in\!(\N\cup\{\bO\})\times 
\bigcup_{k,n\in \Z} (\Z[x_1,\ldots,x_n])^k}$,  
decide whether the complex zero set of 
$F$ has dimension at least $D$. \dia 
\end{quote}  
In particular, {\tt FEAS}$_\C\not\in\pp\Longrightarrow${\tt DIM}$_\C
\not\in\pp$. 
Recall the containment of complexity classes 
$\pp\!\subseteq\!\np\!\subseteq\!\am\!\subseteq\!\pp^{\np^\np}\!
\subseteq\!\pspa$,  
and that the question $\pp\!\stackrel{?}{=}\!\pspa$ remains 
open \cite{papa,cxity}. That $\pp\!=\!\np$ implies the 
{\em collapse} $\pp\!=\!\np\!=\!\conp\!=\!\am\!=\!\pp^{\np^\np}$ 
is a basic fact from complexity theory (see, e.g., \cite[Thm.\ 5.4, pp.\ 
97--98]{cxity}). (We briefly review these complexity classes in the next 
section.) 
{\tt DIM}$_\C$ (and thus {\tt FEAS}$_\C$) has been known to lie in $\pspa$ at 
least since \cite{giustiheintz}, and the 
underlying algorithms have important precursors in 
\cite{chigo,Can88,ierardi,renegar}. 

But in 1996, Koiran \cite{Koi96} proved 
that the truth of the  
Generalized Riemann Hypothesis (GRH) implies that {\tt DIM}$_\C\!\in\!
\am$. In particular, one thus easily obtains that the truth of GRH yields 
the implication {\tt DIM}$_\C\!\not\in\!\pp \Longrightarrow \pp\!\neq\!\np$. 
Thus, assuming GRH, if one can prove that computing the dimension of complex 
algebraic sets is sufficiently hard, then one can solve the $\pp$ vs.\ $\np$ 
Problem. An interesting application of Koiran's result is that it is a key 
step in the proof that the truth of GRH 
implies that knottedness (of a curve defined by a knot diagram) 
can be decided in $\np$ \cite{kuperberg}. 

Here, we prove that {\tt DIM}$_\C\!\in\!\pp^{\np^\np}$ under either of 
two new 
hypotheses: See Theorem \ref{thm:main} below. 
Each of our hypotheses is implied by GRH, {\em but can still hold true under 
certain failures of GRH}. 
\begin{rem} 
To the best of our knowledge, the only other work on improving Koiran's 
conditional speed-up has focussed on proving unconditional speed-ups (from 
$\pspa$ to $\pp^{\np^\np}$ or $\np$) for special 
families of polynomial systems. See, e.g., \cite{cheng,rojas}. For 
instance, thanks to the first paper, the special case of {\tt FEAS}$_\C$ 
involving inputs of the form $(f,x^D_1-1)$ with $(D,f)\!\in\!\N\times \Z[x_1]$ 
is $\np$-complete. \dia 
\end{rem} 

To state our first and most plausible hypothesis, 
let $f\!\in\!\Z[x_1]$ be an irreducible 
polynomial of degree $d$ with coefficients of absolute value at most 
$2^\sigma$ for some 
$\sigma\!\in\!\N$. Let $\pi_f(x)$ denote the number of primes $p$ 
for which the mod $p$ reduction of $f$ 
has a root mod $p$ and $p\!\leq\!x$. 
Note that $\pi_{0}(x)$ is thus simply the number
of primes $p\!\leq\!x$, i.e., the well-known prime-counting function 
$\pi(x)$. In what follows, {\em all} $O$- and $\Omega$- constants 
are absolute (i.e., they really are constants) and effectively computable. 

\begin{MRH}{\em                                         
There is a constant $C\!>\!1$ such that for any $f$ as above we have\\ 
\mbox{}\hfill  
$\pi_f(x) \geq x\left(\frac{1}{d\log x} \ 
  - \ \frac{1}{\exp\!\left(\frac{(\log x)^{1/C}}{(\log(d^2\sigma+d^3))^C}
   \right)}\right)$.\hfill\mbox{}\\ 
for $x=\Omega\!\left(\exp\!\left(4(\log(d^2\sigma+d^3))^{2+C^2}\right)
\right)$.}
\end{MRH}

\noindent 
That $\pi_f(x)$ is asymptotic to $\frac{x}{s_f\log x}$ for 
some positive integer $s_f\!\leq\!d$ goes back to classical 
work of Frobenius \cite{frobenius} (see also \cite{lenste} for 
an excellent historical discussion). More to the point, 
as we'll see in our proofs, the behavior of $\pi_f$ is intimately related 
to the distribution of prime ideals in the ring $\cO_K$ of algebraic integers 
in the number field $K\!:=\!\Q[x_1]/\langle f\rangle$, and the error term is 
where all the difficulty enters: MRH is not currently 
known to be true. However, MRH can still hold even if GRH fails (see Theorem 
\ref{thm:main} below). In particular, 
while the truth of GRH implies that the $1/$exponential term in our lower 
bound above can be decreased to $O\!\left(\frac{d\log(\Delta x)}{\sqrt{x}}
\right)$ in absolute value, we will see later that our looser bound 
still suffices for our algorithmic purposes. (Note that 
$\frac{1}{\sqrt{x}}\!=\!o\!\left(\frac{1}{\exp((\log x)^{1/C})}\right)$ for any 
$C\!>\!1$.)  

Our second hypothesis is a statement intermediate 
between MRH and GRH in plausibility. Recall that the 
{\em Dedekind zeta function}, $\zeta_K(s)$, is the analytic continuation 
(to $\C\setminus\{1\}$) of the  
function $\sum\limits_{\fa} \frac{1}{(N\fa)^s}$, where the summation is taken 
over all integral ideals $\fa$ of $\cO_K$ and $N\fa$ is the norm of $\fa$ 
\cite{ik}.  (So $\zeta_\Q(s)$ is the classical {\em Riemann zeta function 
$\zeta(s)$}, defined from the sum $\sum\limits^\infty_{n=1} \frac{1}{n^s}$.)   
We call a root $\rho = \beta+\gamma\sqrt{-1}$ (with $\beta,\gamma\!\in\!\R$) 
of $\zeta_K$ a \textbf{non-trivial zero} if and only if $0<\beta<1$. GRH is 
then following statement:
\begin{quote}
(GRH) All the non-trivial zeroes $\rho = \beta+\gamma\sqrt{-1}$ of 
$\zeta_K$ lie on 
the vertical line defined by $\beta = 1/2$. \dia 
\end{quote} 

\noindent 
Let $\Delta$ denote the absolute value of the discriminant of $K$. 
Our second hypothesis allows {\em infinitely many} zeroes 
off the line $\beta\!=\frac{1}{2}$, provided they don't approach the 
boundary of the critical region too quickly (as a function of 
$(d,\Delta)$). We review the number theory we 
need in the next section. 
\begin{MDZH}{\em 
There is a constant $C\!>\!4$ such that for any number field $K$,  
the Dedekind zeta function $\zeta_K(s)$ has {\em no} zeroes 
$\rho = \beta+\gamma\sqrt{-1}$ in the region
\begin{align*}
|\gamma| &\geq (1+4\log \Delta)^{-1}\\
\beta &\geq 1 - (\log(d\log(3\Delta))^{C}
\log(|\gamma|+2))^{-1}
\end{align*}
and {\em no} real zeroes in the open interval $\left(
1-\log(d\log(3\Delta))^{-C},1\right)$.}  
\end{MDZH} 

The main motivation for our two preceding hypotheses is the 
following chain of implications, which form our main result. 
\begin{thm} 
\label{thm:main} 
The following three implications hold: \\ 
(1) GRH$\Longrightarrow$MDZH \ \ , \ \ 
(2) MDZH$\Longrightarrow$MRH \ \ , \ \  
(3) MRH$\Longrightarrow${\tt DIM}$_\C\in\!\pp^{\np^\np}$.  
\end{thm} 

\noindent 
We prove Theorem \ref{thm:main} in Section \ref{sec:proof}. We briefly 
review some complexity theoretic notation in 
Section \ref{sub:cxity}, and in Section \ref{sub:alg} we review some 
algebraic tools we need to relate polynomial systems to number fields. 
It is important to recall that,  
like Koiran's original approach in \cite{Koi96}, our algorithm is completely 
distinct from numerical continuation, or the usual computational algebra 
techniques like Gr\"{o}bner bases, resultants, or non-Archimedean Newton 
Iteration. In particular, we use random sampling to study the density of 
primes $p$ for which the mod $p$ reduction of a polynomial system has 
roots over the finite field $\F_p$. 

\section{Technical Background}
Our approach begins by naturally associating a number field $K$ to a  
polynomial system $F\!=\!(f_1,\ldots,f_k)\!\in\!\Z[x_1,\ldots,x_n]$. 
Then, the distribution of prime ideals of $\cO_K$ 
forces the existence of complex roots for $F$ to imply (unconditionally) 
the existence of roots over $\F_p$ for a {\em positive density} of mod $p$ 
reductions of $F$. Conversely, if $F$ has no complex roots, then there 
are (unconditionally) only finitely many primes $p$ such that the mod $p$ 
reduction of $F$ has a root over $\F_p$. These observations, along with a 
clever random-sampling trick that formed the first algorithm for computing 
complex dimension in the polynomial-hierarchy (assuming GRH), go back to 
Koiran \cite{Koi96}. Our key contribution is thus isolating the minimal 
number-theoretic hypotheses (``strictly'' more plausible than GRH) 
sufficient to make a positive density of primes 
observable via efficient random sampling. 

\subsection{Some Complexity Theory}
\label{sub:cxity} 
Our underlying computational model will be the classical Turing machine, 
which, informally, can be assumed to be anyone's laptop computer,  
augmented with infinite memory and a flawless operating system. Our 
notion of input size is the following: 
\begin{dfn}
The \textbf{bit-size} (or \textbf{sparse size}) of a polynomial system 
$F:= (f_1,\cdots,f_k)\!\in\!\Z[x_1,\ldots,x_n]$, is defined to be the 
total number of bits in 
the binary expansions of all the coefficients and exponents of the monomial 
term expansions of all the $f_i$. 
\end{dfn}

Recall that an \textbf{oracle in A} is a special machine that runs, in unit 
time, an algorithm with complexity in A. Our complexity classes can then be 
summarized as follows (and found properly defined in \cite{papa,cxity}).\\
\begin{itemize}
\item[\textbf{P}]{The family of decision problems which can be done within time polynomial in the input size.} 
\item[\textbf{NP}]{The family of decision problems where a ``\texttt{yes}'' answer can be \textbf{certified} within time polynomial in the input size.} 
\item[\textbf{\#P}]{The family of enumerative problems $\mathcal{P}$ admitting an \textbf{NP} problem $\mathcal{Q}$ such that the answer to every instance of $\mathcal{P}$ is exactly the number of ``\texttt{yes}'' instances of $\mathcal{Q}$.} 
\item[$\mathbf{NP^{NP}}$]{The family of decision problems polynomial-time  
equivalent to deciding quantified Boolean sentences of the form\\ 
$\exists x_1 \cdots \exists x_\ell \forall y_1 \cdots \forall y_m \ \ 
B(x_1,\ldots,x_\ell,y_1,\ldots,y_m)$.}  
\item[$\mathbf{P^{NP^{NP}}}$]{The family of decision problems solvable within time polynomial in the input size, with as many calls to an $\mathbf{NP^{NP}}$-oracle as allowed by the time bound.} 
\item[\textbf{PSPACE}]{The family of decision problems solvable within time polynomial in the input size, provided a number of processors exponential in the input size is allowed.} 
\end{itemize}

Finally, let us recall the following important approximation result of 
Stockmeyer. 
\begin{thm}[\cite{Sto85}]
\label{thm:stock}
Any enumerative problem $\cE$ in \textbf{\#P} admits an algorithm in 
$\mathbf{P^{NP^{NP}}}$ which decides if the output of an instance of 
$\cE$ exceeds an input $M\!\in\!\N$ by a factor of $2$. 
\end{thm}

\noindent 
One can thus, in the preceding decisional sense, do constant-factor 
approximation of functions in $\#\pp$ within the polynomial-hierarchy. 

\subsection{Rational Univariate Reduction and an Arithmetic Nullstellensatz}
\label{sub:alg}
In this section, we develop tools that will reduce the feasibility of 
polynomial systems to algebra involving ``large'' univariate polynomials. 
The resulting quantitative bounds are essential in constructing our algorithm. 

Our first lemma is a slight refinement of earlier work on rational univariate 
reduction (see, e.g., \cite{Can88,rojas0,maillot}), so we leave its proof 
for the full version of this paper. 
\begin{lem} 
\label{lemma:1} 
	Let $F$ be our polynomial system and $Z_F$ denote the zero set of $F$ 
in $\C^n$. Then there are univariate polynomials $u_1,\cdots,u_n, U_F\in 
\Z[t]$ and positive integers $r_1,\cdots,r_n$ such that
	\begin{enumerate}
		\item The number of irreducible components of $Z_F$ is bounded above by $\deg U_F$, and $\deg(u_i)\leq \deg(U_F) \leq D^n$ for all $1\leq i\leq n$.
		\item For any root $\theta$ of $U_F$, we have 
$F\!\left(\frac{u_1(\theta)}{r_1},\cdots,\frac{u_n(\theta)}{r_n}\right) = 0$, 
and every irreducible component of $Z_F$ contains at least one point that can be expressed in this way.
		\item The coefficients of $U_F$ have absolute value 
no greater than $2^{O(D^n[\sigma(F)+n\log D])}$. \qed 
	\end{enumerate}

\end{lem}

If $F$ has finitely many roots then $(U_F,u_1,r_1,\ldots,u_n,r_n)$ 
will capture all the roots of $F$ in the sense above. 
Let $f$ be the square-free part of $U_F$. Note that if $p\nmid \lcm(r_1,\cdots,r_n)$ and $f$ mod $p$ has a root, then $F$ mod $p$ also has a root. 

Now consider the following recently refined effective arithmetic version of 
Hilbert's Nullstellensatz. Recall that the \textbf{height} of a polynomial $f$, denoted by $h(f)$, is defined as the logarithm of the maximum of the absolute value of its coefficients. 
\begin{prop}[\cite{DKS13}] 
\label{prop:1} 
Let $D= \max_i \deg(f_i)$, and $h = \max_i h(f_i)$. Then the polynomial system $F$ has no roots in $\C^n$ if and only if there exist polynomials $g_1,\cdots, g_k\in \Z[x_1,\cdots,x_n]$ and a positive integer $\alpha$ satisfying the Bez\'out identity $f_1g_1+\cdots+f_kg_k = \alpha$, and
\begin{enumerate}
\item $\deg(g_i) \leq 4nD^n$,
\item $h(\alpha), h(g_i) \leq 4n(n+1)D^n(h+\log k + (n+7)\log(n+1)D)$.
\end{enumerate}
\end{prop}

If the mod $p$ reduction of $F$ has a root over $\F_p$ , then $p$ divides $\alpha$. There are at most $1+\log\alpha$ many prime factors of an integer $\alpha$, hence 

\begin{thm}
\label{thm:noroot} 
If $F$ has no complex root then the mod $p$ reduction of $F$ has a root 
over $\F_p$ for no more than $A_F$ primes $p$, where
	\begin{align*}
	A_F = 4n(n+1)D^n(h+\log k + (n+7)\log(n+1)D). \qed 
	\end{align*}
\end{thm}

If we can somehow certify that the mod $p$ reduction of $F$ has roots 
over $\F_p$ for at least $A_F+1$ many primes $p$, then we can certify 
that $F$ has complex roots.  

\begin{ex}
The following system $F$ of two univariate polynomials:
{\tiny\begin{align*}
f_1&=x^{120017}+ 4x^{110001}+ 19x^{110000}- 3x^{101208}+ x^{100000}- 47x^{25018}+ 37x^{20017} \\
&-188x^{15002}- 893x^{15001}+ 148x^{10001}+ 703x^{10000}+ 141x^{6209}-47x^{5001}- 111 x^{1208}+ 37\\
f_2&=19x^{210017}+ 76x^{200001}+ 361x^{200000}- 57x^{191208}+ 19x^{190000}+ 2x^{30016}- 7x^{20017} \\
&+8x^{20000}+38x^{19999}- 6x^{11207}- 28x^{10001}- 133x^{1000}+2x^{9999}+21 x^{1208}-7,
\end{align*}}
has a complex root. It is easy to compute that $A_F \approx1.9567\times 10^{12}$. However, as it is a small system, we can get a better bound on the Bez\'out constant $\alpha$ by computing the determinant of the corresponding Sylvester matrix. Moreover, we can use a finer result due to Robin (\cite{Rob83}) on 
$\omega(\alpha)$ (the number of prime factors of $\alpha$). 
\begin{align*}
 \omega(\alpha) < \frac{\log\alpha}{\log\log\alpha} + \frac{\log\alpha}{(\log\log\alpha)^2}+2.89726\frac{\log\alpha}{(\log\log\alpha)^3},
\end{align*}
for $\alpha \geq 3$. Therefore, to determine if $F$ has a $\C$ root, it suffices to check if the number of primes $p$ such that the mod $p$ reduction of $F$ 
has a root in $\F_p$ is more than $163,317$. In fact, such $p$ comprise roughly $2/3$ of the first $163,317$ primes. \dia 
\end{ex}

\subsection{Prime Ideals}
\label{sub:num} 
In what follows, $p$ always denotes a prime in $\N$, and $\fp$ a 
prime ideal in the number ring $\cO_K$. 

For any number field $K$, let $\pi_K(x)$ denote the number of $\fp$ satisfying 
$N\fp \leq x$. Recall that the ideal norm is defined to be 
$N\fa:= |\cO_K/\fa|$.  The classical Prime Ideal Theorem \cite{ik} 
states that for any number field $K$, $\pi_K(x)$ is asymptotic to 
$\frac{x}{\log x}$. 

Let $\pi_F(x)$ be the number of primes $p$ such that the mod $p$ reduction 
of $F$ has a root over $\F_p$ and $p\!\leq\!x$. (So our earlier $\pi_f$ was the 
univariate version of $\pi_F$.) 
The main idea behind proving that GRH implies MDZH is   
an approximation,  
with an explicit error term, of the weighted prime-power-counting function 
$\psi_K(x)$ associated to $\pi_K(x)$, defined by
\begin{align*}
	\psi_K(x) = \sum\limits_\fp \log N\fp. 
\end{align*}
Here the sum is taken over the unramified primes such that $N\fp^m \leq x$ for some $m$. To start, we first quote the following important lemmas:

\begin{lem}[\cite{LO77}, Lemma 7.1] 
\label{lemma:2} 
Let $\rho = \beta+\gamma\sqrt{-1}$ denote a non-trivial zero of $\zeta_K$ 
(so $0<\beta<1$). For $x\!\geq\!2$, and $T\!\geq\!2$, define
\begin{align*}
S(x,T) =\sum_{|\gamma|<T} \frac{x^\rho}{\rho}- \sum_{|\rho|<\frac{1}{2}}\frac{1}{\rho}.
\end{align*}
Then 
\begin{align*}
\psi_K(x) - x + S(x, T)\ll& \frac{x\log x + T}{T}\log \Delta + d\log x + \frac{dx\log x\log T}{T}\\
& + \log x \log \Delta + dx T^{-1}(\log x)^2.
\end{align*}
\end{lem}

\begin{lem}\cite[Proof of Thm.\ 9.2]{LO77}  
\label{lemma:3} 
	Using the notation above,
	\begin{enumerate} 
		\item $\zeta_K$ has at most one non-trivial zero $\rho$ in 
the region
		\begin{align*}
			|\gamma| &\leq (4\log \Delta)^{-1}\\
			\beta &\geq 1 - (4\log\Delta)^{-1}.
		\end{align*} 
		This zero, if it exists, has to be real and simple.  If it 
exists and we call it $\beta_0$ then it must satisfy 
		\begin{align*}
		\frac{x^{1-\beta_0}}{1-\beta_0} + \frac{1}{1-\beta_0} = x^{\sigma}\log x \leq x^{1/2} \log x
		\end{align*}
		for some $0 \leq \sigma\leq 1-\beta_0$.
		\item For $\rho \neq \beta_0$, we have
		\begin{align*}
		\sum_{\substack{\rho \neq 1-\beta_0\\ |\rho|<\frac{1}{2}}}
\left(\frac{x^\rho}{\rho}- \frac{1}{\rho}\right) 
\ll x^{1/2}\sum_{\substack{\rho \neq 1-\beta_0\\ |\rho|<\frac{1}{2}}}\left|
\frac{1}{\rho}\right|\ll x^{1/2}(\log \Delta)^2
		\end{align*}
		\item If we have further that $T\geq 2$ then  
		\begin{align*}
		\sum_{\substack{|\rho| \geq \frac{1}{2} \\ |\gamma|<T}} \left|
\frac{1}{\rho}\right| \ll \log T\log(\Delta T^{d})
		\end{align*}
	\end{enumerate}
\end{lem}

\begin{rem}
An earlier unconditional zero-free region is the following (\cite{LO77}):
\begin{align*}
|\gamma| &\geq (1+4\log \Delta)^{-1}\\
\beta &\geq 1 - \eps(\log\Delta + \log(|\gamma|+2))^{-1},
\end{align*}
for some constant $\eps\!>\!0$. It is easily checked that for any fixed $C$ 
(and any sufficiently large $d$ and $\Delta$) the preceding region is strictly 
contained in the zero-free region of MDZH. Unfortunately, the unconditional  
region of \cite{LO77}, and even the best current unconditional refinements, 
are too small  to guarantee that our upcoming algorithm is in the 
polynomial-hierarchy. \dia
\end{rem}
\begin{rem}
We call the $\beta_0$ from Lemma \ref{lemma:2} 
a \textbf{Siegel-Landau zero}. Observe that $\beta_0$ is a potential counterexample to GRH since it is known that 
\begin{align*}
\beta_0\geq 1-(4\log\Delta)^{-1},  
\end{align*} 
and the right-hand side is at least $3/4$ for sufficently large 
$\Delta$, thus contradicting GRH. 
By using Lemma \ref{lemma:3} in the following 
discussion, we take into account the possibility of a Siegel-Landau zero. \dia 
\end{rem}

\begin{prop}
\label{prop:2} 
	Assuming MDZH with constant $C$, there is an effectively computable 
positive function $c_2(C)$ such that if 
	\begin{align*}
	x \geq \exp\!\left(4(\log\log(3\Delta))^2\log(d\log(3\Delta))^{C^2}
\right)
	\end{align*}
	then 
	\begin{align*}
	\psi_K(x) =x -\frac{x^{\beta_0}}{\beta_0} + R(x)
	\end{align*}
	where
	\begin{align*}
	|R(x)| \leq x\exp\!\left(-c_2(C)\frac{(\log x)^{1/C}}
{(\log(d\log(3\Delta)))^C}\right) 
	\end{align*}
	and the term $\frac{x^{\beta_0}}{\beta_0}$ only occurs if $\zeta_K(s)$ has a Siegel-Landau zero $\beta_0$. 
\end{prop}
\begin{proof}
	By simply applying the Lemma \ref{lemma:3}, we have
	\begin{align*}
	&S(x,T) - \frac{x^{\beta_0}}{\beta_0}\leq \frac{x^{1-\beta_0}}{1-\beta_0} + \frac{1}{1-\beta_0}+\sum_{\substack{\rho \neq 1-\beta_0\\ |\rho|<\frac{1}{2}}} \left(\frac{x^\rho}{\rho}- \frac{1}{\rho}\right) + 
\sum_{\substack{|\rho| \geq \frac{1}{2} \\ |\gamma|<T}} \frac{x^\rho}{\rho}\\
	&\ll x^{1/2}\log x + x^{1/2}(\log \Delta)^2 +  \sum_{\substack{|\rho| \geq \frac{1}{2} \\ |\gamma|<T}} \frac{x^\rho}{\rho}\\
	&\ll x^{1/2}\log x + x^{1/2}(\log \Delta)^2 +  \log T\log(\Delta T^{d})\max_{\substack{|\rho| \geq \frac{1}{2} \\ |\gamma|<T}} |x^\rho|.
	\end{align*}
	On the other hand, let $\rho = \beta + i\gamma$ be a non-trivial zero of $\zeta_K(s)$ with $|\gamma| \leq T$, and $\rho$ is not a Siegel-Landau zero. 
As MDZH assumes a zero-free region dependent on a given constant $C$,
	\begin{align*}
	|x^{\rho}| = x^{\beta} \leq x\exp\!\left(-c_3\frac{\log x}
{\log\log(3\Delta)+(\log(d\log(3\Delta)))^{2C}\log T}\right)	
	\end{align*}
	for some constant $c_3$. 
	Now take 
	\begin{align*}
	T = \exp\!\left((\log(d\log(3\Delta)))^{-C}(\log x)^{1-1/C}
        -\log\log(3\Delta)\right).
	\end{align*}
	The estimate of the theorem then follows from the above computation, 
and Lemma \ref{lemma:2}. 
\end{proof}

\section{The Proof of Theorem \ref{thm:main}} 
\label{sec:proof} 
Since GRH trivially implies MDZH, Assertion (1) is tautologically true. 
So we now proceed with proving Assertions (2) and (3). 
\subsection{The Proof of Assertion (2): MDZH $\Longrightarrow$ MRH} 
	Define $\theta_K(x) = \sum \log N\fp$ where the summation is over all the unramified prime $\fp$ such that $N\fp \leq x$. There are at most $d$ ideals $\fp^m$ of a given norm in $K$, hence
    \begin{align*}
    0\leq \psi_K(x) - \theta_K(x) &= \sum_{N\fp^m\leq x, m\geq 2}\log N\fp\\
    &\leq \sum_{m=2}^{\log_2x} dx^{1/m}\log x \leq 3d\sqrt{x}\log x.
    \end{align*}
    The error term $R(x)$ still dominates this discrepancy, so the estimates in Proposition \ref{prop:2} still holds when $\psi_K(x)$ is replaced by $\theta_K(x)$. By a standard partial summation trick we have:
	\begin{align*}
		\left|\pi_K(x) -\frac{x}{\log x}\right|\leq  
\frac{x^{\beta_0}}{\beta_0\log x} + O\!\left(x\exp(-\frac{(\log x)^{1/C}}
{(\log(d\log(3\Delta)))^C})\right),
	\end{align*}
	for 
    \begin{align*}
    x \geq \exp\!\left(4(\log\log(3\Delta))^2\log(d\log(3\Delta))^{C^2}\right).
    \end{align*}
By the last assertion of MDZH, the error arising from the possible existence of the Siegel-Landau zero $\beta_0$ is dominated by $R(x)$ for $x$ 
in the range of $x$ we are using. Therefore,
\begin{align*}
\pi_K(x) \geq x\left(\frac{1}{\log x} - O\!\left(\exp\!\left(-\frac{(\log x)^{1/C}}
{(\log(d \log(3\Delta)))^C}\right)\right)\right).
\end{align*} 
 
Let $W(p)$ be the number of linear factors of $f$ mod $p$. The key fact we 
observe now is that if $p\nmid \Delta$ then $W(p)$ equals the number of prime 
ideals $\fp$ of $K$ of degree $1$ that lie over $p$. Thus, $\sum_{p\leq x} W(p)$ counts the number of $\fp$ of degree $1$ with norm up to $x$. As the prime 
ideals of degree greater than $1$ must lie over a prime number $p\leq x^{1/2}$, and there are no more than $d$ such $p$, we have
    \begin{align*}
    \pi_K(x) - \sum_{p\leq x} W(p) = O(dx^{1/2}).
    \end{align*}
Note that if $p$ divides the discriminant of $f$, the correspondence between $\fp$ of degree $1$, and the linear factors of $f$ mod $p$ will break. But there are no more than $\log\Delta$ such $p$. However, the error term coming from $R(x)$ still dominates. Therefore, $\sum_{p\leq x} W(p)$ satisfies the same estimate as $\pi_K(x)$. 

Let $r(p) = 1$ if $f$ has a root in $\F_p$ and $0$ otherwise. As $f$ is irreducible of degree $d$, so $f$ mod $p$ is non-trivial. Then 
    \begin{align*}
    \pi_f(x) = \sum_{p\leq x} r(p) \geq \sum_{p\leq x} W(p)/d,
    \end{align*}
and MRH thus follows upon recalling that $\log \Delta\!=\!
O(d^2\sigma+d^3)$ \cite{Roj01}.    \qed 

\begin{rem}
We will deal later with square-free polynomial that are possibly reducible. In this case, we write $f(x) = \prod f_i(x)$, with $f_i(x)$ irreducible, and apply the same argument to each summand of
\begin{align*}
\Q[x]/\la f(x)\ra \cong \oplus \Q[x]/\la f_i(x)\ra.
\end{align*}
Therefore, we can replace the ``irreducible'' assumption in MRH with 
``square-free''. \dia 
\end{rem}

\subsection{The Proof of Assertion (3): MRH $\Longrightarrow$
{\tt DIM}$_\C\in\!\pp^{\np^\np}$} 
Let $u_1,\cdots, u_n, U_F\in \Z[t]$ and $r_1,\cdots,r_n$ respectively be the 
polynomials and integers arising from a rational univariate reduction of $F$. 
Let $K = \Q[x]/\la f\ra$, where $f$ is the square-free part of $U_F$. Then 
$d=\deg f\leq D^n$. Moreover, assuming 
the coefficients of $U_F$ have absolute value no greater than $2^{\sigma(F)}$, 
we can 
effectively bound the discriminant of $f$: $\log\Delta = O((\deg U_F)^2\sigma(U_F) + (\deg U_F)^3)$ \cite{Roj01}. Note that if $p\nmid \lcm(r_1,\cdots,r_n)$, 
then Assertion (2) of Lemma \ref{lemma:1} continues to hold modulo $p$. That is, if in addition $f$ has a root in $\F_p$, then $F$ has a root over $\F_p$. 
Hence we have $\pi_F(x) \geq\pi_f(x)$. 

Recall from Theorem \ref{thm:noroot} that if $F$ has no complex solutions, 
then the mod $p$ reduction of $F$ has a root over $\F_p$ for at most $A_F$ 
many primes $p$. On the other hand, we have the following result: 
\begin{prop}
\label{prop:3} 
If $F$ has a complex root then there is a positive function $t(F)$ 
such that $\pi_F(x) \geq 7A_F$ for every $x\geq t(F)$. In particular, 
$\log t(F)$ is polynomial in the bit-size of $F$. 
\end{prop}

\noindent 
{\bf Proof of Proposition \ref{prop:3}:} 
Recall from MRH that the asymptotic formula for $\pi_f(x)$, and thus 
$\pi_F(x)$, only holds for $x$ sufficiently large. In particular, we need 
\begin{align*}
x \geq \exp\!\left(4(\log\log(3\Delta))^2\log(d\log(3\Delta))^{C^2}\right). 
\end{align*} 
Let $t_1$ denote this lower bound and let $\sigma(F)$ denote the 
bit-size of $F$. It is easy to see that for $C \geq 2$, 
\begin{align*}
\log t_1 &\leq  O\!\left(\log(D^{3n}(\sigma(F)+n\log D))^{C^2}\right) \\
&= O\!\left((3\sigma(F)^2+2\log\sigma(F))^{C^2}\right) = 
O\!\left(\sigma(F)^{4C^2}\right),
\end{align*}
which is polynomial in $\sigma(F)$. 

On the other hand, by applying the numerical bounds from Lemma \ref{lemma:1}, 
and MRH with constant $C$, we see that $\pi_F(x) \geq 7A_F$ if:
	\begin{align*}
 & x\left(\frac{1}{d\log x} - \exp\!\left(-\frac{(\log x)^{1/C}}
{(\log(d\log(3\Delta)))^C}\right)\right) & \\
    &\geq 28n(n+1)D^n(h+\log k + (n+7)\log(n+1)D). & 
    \end{align*}
    Necessarily,
    \begin{align*}
	\frac{1}{D^n}\frac{x}{\log x} &\gg x\exp(-\frac{(\log x)^{1/C}}{(\log(D^n\log\Delta))^C})(n+k)^2D^{n+1}(\sigma(F) + n\log n)). 
	\end{align*}
	Now with $\log \Delta = O(\deg U_F\sigma(U_F) + d^2) = O(D^{2n}(\sigma(F)+n\log D))$, and $n\log D\leq (n+\log D)^2\leq \sigma(F)^2$, we have
	\begin{align*}
	\Leftarrow &\frac{x}{\log x} \gg x\exp(-\frac{(\log x)^{1/C}}{(2\log(D^{3n}\sigma(F)^2)^C})(n+k)^2D^{2n+1}\sigma(F)^2, \\
	\Leftarrow &\log x \gg \log\log x+\log x - \frac{(\log x)^{1/C}}{(2\log(D^{3n}\sigma(F)^2)^C}+7\sigma(F)^2,\\
	\Leftarrow &  \frac{(\log x)^{1/C}}{(6\sigma(F))^{2C}} \gg \log\log x+7\sigma(F)^2,
	\end{align*}
	which holds if $\log x \geq \log t_2 := O(\sigma(F)^{4C^2})$. The proposition follows by letting $t(F):=\max(t_1,t_2)$. \qed

Continuing our proof of Assertion (3) of Theorem  \ref{thm:main}, consider 
the following algorithm: 
 
\begin{tabular}{ p{1.5cm} p{13.5cm} }
\multicolumn{2}{l}{\large{\texttt{PHFEAS}}} \\
\textbf{Input} & A $k\times n$ polynomial system $F$ with integer coefficients.\\
\textbf{Output} & A true declaration whether $F$ has a complex root.\\
\textbf{Step 1} & Compute $A_F$ and $t(F)$ from Theorem \ref{thm:noroot} and 
Proposition \ref{prop:3}. \\
\textbf{Step 2} & Use Stockmeyer's algorithm as in Theorem \ref{thm:stock}, to approximate the number $M$ of primes $p\in\{1,\cdots, t(F)\}$, such that the mod $p$ reduction of $F$ has a root over $\F_p$. \\
\textbf{Step 3} & If $M> 3A_F$, then declare that $F$ has a complex root. Otherwise, declare that $F$ has no complex root.
\end{tabular}

Since {\tt DIM}$_\C$ can be reduced in $\bpp$ to {\tt FEAS}$_\C$, and 
$\bpp\!\subseteq\!\am$ \cite{papa,cxity}, 
it suffices to prove that algorithm {\tt PHFEAS} is correct and runs in time 
$\mathbf{P^{NP^{NP}}}$. 

Toward this end, observe that if $F$ has no complex root, then there are no 
more than $A_F$ primes $p$ such that the mod $p$ reduction of $F$ has a 
root over $\F_p$. By Proposition 
\ref{prop:3} and assuming MRH, if $F$ has a 
complex root, then there are at least $7A_F$ primes $p\leq t(F)$ such that $F$ 
mod $p$ has a root. Such primes have bit-size no greater than $O(\log t(F))$. 
It is 
also easy to check that $\log A_F$ is also polynomial in $\sigma(F)$. 
Moreover, primality checking can be done in $\mathbf{P}$, and the existence of 
roots of $F$ over $\F_p$ can be done in $\mathbf{NP}$. Hence the number of 
primes we are approximating is computable in $\textbf{\#P}$. So the algorithm 
is correct and runs in time $\mathbf{P^{NP^{NP}}}$. \qed 

\bibliographystyle{ACM-Reference-Format}

\end{document}